\documentclass[fleqn,11pt,twoside]{article}

\usepackage{amsthm,amsthm,amssymb, color, xcolor,epsfig, graphics, subfigure}

\usepackage{amsmath, graphicx, latexsym, lscape}

\makeatletter
\newcommand{\copyrightnote}[2]{{\renewcommand{\thefootnote}{}
 \footnotetext{\small\it
\begin{flushleft}
 \copyright \ #1   #2  
\end{flushleft}}}}

\newcommand{\Name}[1]{\begin{flushleft}
                       \LARGE \bf #1
                       \end{flushleft}\vspace{-3mm}}

\newcommand{\Author}[1]{\begin{flushleft}
                       \it #1 \end{flushleft}}

\newcommand{\Address}[1]{\begin{flushleft}
                       \it #1 \end{flushleft}}

\newcommand{\Date}[1]{\begin{flushleft}
                      \small  \it #1 \end{flushleft}}

\newcommand{\evenhead}{Author \ name}
\newcommand{\oddhead}{Article \ name}

\renewcommand{\@evenhead}{
\hspace*{-3pt}\raisebox{-15pt}[\headheight][0pt]{\vbox{\hbox to \textwidth
{\thepage \hfil \evenhead}\vskip4pt \hrule}}}
\renewcommand{\@oddhead}{
\hspace*{-3pt}\raisebox{-15pt}[\headheight][0pt]{\vbox{\hbox to \textwidth
{\oddhead \hfil \thepage}\vskip4pt\hrule}}}
\renewcommand{\@evenfoot}{}
\renewcommand{\@oddfoot}{}

\long\def\@makecaption#1#2{%
  \vskip\abovecaptionskip
  \sbox\@tempboxa{\small \textbf{#1.}\ \ #2}%
  \ifdim \wd\@tempboxa >\hsize
    {\small \textbf{#1.}\ \ #2}\par
  \else
    \global \@minipagefalse
    \hb@xt@\hsize{\hfil\box\@tempboxa\hfil}%
  \fi
  \vskip\belowcaptionskip}

\newcommand{\JNMPnumberwithin}[3][\arabic]{%
  \@ifundefined{c@#2}{\@nocounterr{#2}}{%
    \@ifundefined{c@#3}{\@nocnterr{#3}}{%
      \@addtoreset{#2}{#3}%
      \@xp\xdef\csname the#2\endcsname{%
        \@xp\@nx\csname the#3\endcsname .\@nx#1{#2}}}}%
}

\renewenvironment{proof}[1][\proofname]{\par
  \normalfont
  \topsep6\p@\@plus6\p@ \trivlist
  \item[\hskip\labelsep\textbf{%
    #1\@addpunct{.}}]\ignorespaces
}{%
  \qed\endtrivlist
}

\newcommand{\resetfootnoterule} {
  \renewcommand\footnoterule{%
  \kern-3\p@
  \hrule\@width.4\columnwidth
  \kern2.6\p@}
}

\renewcommand{\footnoterule}{}

\newcommand{\bbR}{{\mathbb R}}

\makeatother

\theoremstyle{definition}

\newtheorem{theorem}{Theorem}

\begin{document}

\renewcommand{\evenhead}{ {\LARGE\textcolor{blue!10!black!40!green}{{\sf \ \ \ ]ocnmp[}}}\strut\hfill Yuri B Suris}
\renewcommand{\oddhead}{ {\LARGE\textcolor{blue!10!black!40!green}{{\sf ]ocnmp[}}}\ \ \ \ \ Integrals of discretizations by polarization}

\thispagestyle{empty}
\newcommand{\FistPageHead}[3]{
\begin{flushleft}
\raisebox{8mm}[0pt][0pt]
{\footnotesize \sf
\parbox{150mm}{{Open Communications in Nonlinear Mathematical Physics}\ \ \ \ {\LARGE\textcolor{blue!10!black!40!green}{]ocnmp[}}
\quad Special Issue 1, 2024\ \  pp
#2\hfill {\sc #3}}}\vspace{-13mm}
\end{flushleft}}

\FistPageHead{1}{\pageref{firstpage}--\pageref{lastpage}}{ \ \ }

\strut\hfill

\strut\hfill

\copyrightnote{The author(s). Distributed under a Creative Commons Attribution 4.0 International License}

\begin{center}
{  {\bf This article is part of an OCNMP Special Issue\\ 
\smallskip
in Memory of Professor Decio Levi}}
\end{center}

\smallskip

\Name{A new approach to integrals of discretizations by polarization}

\Author{Yuri B. Suris}

\Address{Institut f\"ur Mathematik, MA 7-1,
Technische Universit\"at Berlin, Str. des 17. Juni 136,
10623 Berlin, Germany.
E-mail: {\tt  suris@math.tu-berlin.de}}

\Date{Received July 12, 2023; Accepted October 6, 2023}

\setcounter{equation}{0}

\begin{abstract}

\noindent 
Recently, a family of unconventional integrators for ODEs with polynomial vector fields was proposed, based on the polarization of vector fields. The simplest instance is the by now famous Kahan discretization for quadratic vector fields. All these integrators seem to possess remarkable conservation properties. In particular, it has been proved that, when the underlying ODE is Hamiltonian, its polarization discretization possesses an integral of motion and an invariant volume form. In this note, we propose a new algebraic approach to derivation of the integrals of motion for  polarization discretizations.

\end{abstract}

\label{firstpage}


\section{Introduction}

The by now famous Kahan discretization \cite{K} is a one-step numerical method designed specially for ODEs in $\bbR^d$,
\begin{equation}\label{eq dyn syst}
\dot x=f(x),
\end{equation}
with all components of the vector field $f$ being polynomials of degree 2. The Kahan discretization with the stepsize $\epsilon$ is the following difference equation:
\begin{equation}\label{eq Kahan}
(x_{n+1}-x_n)/\epsilon ={\rm pol}_2f(x_n,x_{n+1}).
\end{equation}
Here, for any quadratic form $Q(x)$ on $\bbR^d$, its polarization is the corresponding symmetric bilinear form,
$$
{\rm pol}_2Q(x_1,x_2)=\frac{1}{2}\big(Q(x_1+x_2)-Q(x_1)-Q(x_2)\big).
$$
For a non-homogeneous polynomial $P(x)$ of degree 2, one extends it to a quadratic form in homogeneous coordinates, $\widetilde P(x,z)=z^2 P(x/z)$, computes the bilinear symmetric form ${\rm pol}_2\widetilde P$ and then sets ${\rm pol}_2P={\rm pol}_2\widetilde P|_{z_1=z_2=1}$. In particular, for a linear form $L(x)$ we obtain ${\rm pol}_2L(x_1,x_2)=(L(x_1)+L(x_2))/2$.

Equation \eqref{eq Kahan} is linear with respect to $x_{n+1}$, thus can be solved to give a rational map
\begin{equation}\label{eq Kahan map}
x_{n+1}=f_\epsilon(x_n).
\end{equation}
Moreover, due to the symmetry of equation \eqref{eq Kahan} with respect to $x_n\leftrightarrow x_{n+1}$ and $\epsilon\leftrightarrow -\epsilon$, we have $f_\epsilon^{-1}=f_{-\epsilon}$, in particular, $f_\epsilon$ is a birational map.

Kahan's discretization is known to inherit integrals and integral invariants much more frequently than could be anticipated, see \cite{PPS1, PPS2, CMOQ1, CMOQ2} and a more recent literature. In the present note, we will address the remarkable result of \cite{CMOQ1} which states that map $f_\epsilon$ always possesses an integral of motion, if $f(x)$ is a quadratic {\em Hamiltonian} vector field in the space of an even dimension $d$, that is, $f(x)=J\nabla H(x)$, where $H$ is a polynomial of degree 3, and $J\in {\rm so}(d)$ is a non-degenerate skew-symmetric matrix. Our Theorem \ref{th d Ham system int} in Section \ref{sect Ham Kahan} gives a novel derivation and an algebraic interpretation of this result.
\smallskip

A wide generalization of the Kahan discretization for polynomial vector fields of higher degrees was proposed in \cite{CMOQ3}. The most interesting version of this approach deals with higher order differential equations, for which the discretization preserves the dimension of the phase space \cite{HQ}. Consider a second order differential equation in $\bbR^d$,
\begin{equation}\label{eq 2nd order}
\ddot x =g(x),
\end{equation}
where all components of $g(x)$ are polynomials of degree 3. The polarization discretization of such an equation with the stepsize $\epsilon$ is the following second order difference equation:
\begin{equation}\label{eq 2nd order discrete}
(x_{n+1}-2x_n+x_{n-1})/\epsilon^2={\rm pol}_3g(x_{n-1},x_n,x_{n+1}).
\end{equation}
Here, the third order polarization ${\rm pol}_3$ for a cubic form $C(x)$ is the corresponding symmetric trilinear form 
\begin{eqnarray*}
{\rm pol}_3C(x_1,x_2,x_3) & = & \frac{1}{6}\big(C(x_1+x_2+x_3)-C(x_1+x_2)-C(x_1+x_3)-C(x_2+x_3)\\
 & & +C(x_1)+C(x_2)+C(x_3)\big).
\end{eqnarray*}
For a non-homogeneous polynomial $P(x)$ of degree 3, one first extends it to a cubic form in homogeneous coordinates, $\widetilde P(x,z)=z^3 P(x/z)$, computes the trilinear symmetric form ${\rm pol}_3\widetilde P$, and then sets ${\rm pol}_3P={\rm pol}_3\widetilde P|_{z_1=z_2=z_3=1}$. For a quadratic form $Q(x)$, we find:
\begin{equation}\label{eq pol3 Q}
{\rm pol}_3Q(x_1,x_2,x_3)=\frac{1}{3}\big({\rm pol}_2Q(x_1,x_2)+{\rm pol}_2Q(x_1,x_3)+{\rm pol}_2Q(x_2,x_3)\big),
\end{equation}
while for a linear form $L(x)$, we find:
\begin{equation}\label{eq pol3 L}
{\rm pol}_3L(x_1,x_2,x_3)=\frac{1}{3}\big(L(x_1)+L(x_2)+L(x_3)\big).
\end{equation}

Again, equation \eqref{eq 2nd order discrete} is linear with respect to $x_{n+1}$, thus can be solved to give a birational map
\begin{equation}\label{eq polarization map}
(x_n,x_{n+1})=g_\epsilon(x_{n-1},x_n),
\end{equation}
which enjoys the symmetry with respect to $x_{n-1}\leftrightarrow x_{n+1}$. 

Let $g(x)=K\nabla W(x)$, where $K\in {\rm Symm}(d)$ is a non-degenerate symmetric matrix, and $W(x)$ is a polynomial of degree 4. Then equation \eqref{eq 2nd order} is equivalent to a canonical Hamiltonian system with the Hamilton function $H(x,p)=\frac{1}{2}\langle p, Kp\rangle+W(x)$. Indeed, equations of motion of the latter read $\dot x=Kp$, $\dot p=-\nabla W(x)$.  A remarkable result of \cite{HQ, MMQ} states that in this case map $g_\epsilon$ possesses an integral of motion. Our Theorem \ref{th Ham syst 2nd order}  in Section \ref{sect 2nd order} gives a novel derivation and algebraic interpretation of this result.

\section{Hamiltonian systems with a cubic integral}
\label{sect Ham Kahan}

Consider a Hamiltonian system
\begin{equation}\label{eq Ham system}
\dot{x}=J\nabla H, 
\end{equation}
where $H:\bbR^{d}\to \bbR$ is a polynomial of degree 3, and $J\in\text{so}(d)$ is a non-degenerate matrix (so that $d$ is necessarily even). It is well known that $H(x)$ is an integral of motion for \eqref{eq Ham system}. We consider the Kahan discretization for \eqref{eq Ham system}, see equation \eqref{eq Kahan}.

\begin{theorem}\label{th d Ham system int}
Separate $H$ into homogeneous parts of degrees 3, 2, and 1:
\begin{equation}
H(x)=H_3(x)+H_2(x)+H_1(x).
\end{equation}
Then the following quantity is a conserved quantity for the difference equation \eqref{eq Kahan}:
\begin{equation}\label{eq d Ham int 1}
H_\epsilon(x_n,x_{n+1})  = \tfrac{1}{ \epsilon}\langle x_n\,,J^{-1}x_{n+1}\rangle +{\rm pol}_2H_2(x_n,x_{n+1})+2\,{\rm pol}_2H_1(x_n,x_{n+1}).
\end{equation}
\end{theorem}
\begin{proof} 
We are dealing with the following difference equation:
\begin{equation}\label{eq Ham Kahan sep}
J^{-1}(x_{n+1}-x_n)/\epsilon={\rm pol}_2 \nabla H_3(x_n,x_{n+1})+{\rm pol}_2 \nabla H_2(x_n,x_{n+1})+{\rm pol}_2 \nabla H_1(x_n,x_{n+1}).
\end{equation}
Take the scalar product of equation \eqref{eq Ham Kahan sep} with $x_{n-1}$:
\begin{eqnarray}\label{eq aux Ham 1}
\lefteqn{\tfrac{1}{\epsilon}\langle x_{n-1}, J^{-1}x_{n+1}\rangle - \tfrac{1}{\epsilon}\langle x_{n-1}, J^{-1}x_{n}\rangle}\nonumber \\
 & = & 3\,{\rm pol}_3 H_3(x_{n-1},x_n,x_{n+1}) +{\rm pol}_2 H_2(x_{n-1},x_n)+{\rm pol}_2 H_2(x_{n-1},x_{n+1}) +H_1(x_{n-1}). 
 \qquad 
\end{eqnarray}
Here we used Euler's theorem on homogeneous functions and have taken into account that for a quadratic form $H_2$ there holds ${\rm pol}_2 \nabla H_2(x_n,x_{n+1})=(\nabla H_2(x_n)+\nabla H_2(x_{n+1}))/2$, and for a linear form $H_1$ its gradient $\nabla H_1$ is a constant vector.
Similarly, take the scalar product of the downshifted (i.e., $n\to n-1$) equation  \eqref{eq Ham Kahan sep} with $x_{n+1}$:
\begin{eqnarray}\label{eq aux Ham 2}
\lefteqn{\tfrac{1}{\epsilon}\langle x_{n+1}, J^{-1}x_n\rangle - \tfrac{1}{\epsilon}\langle x_{n+1}, J^{-1}x_{n-1}\rangle}\nonumber \\
 & = & 3\,{\rm pol}_3 H_3(x_{n-1},x_n,x_{n+1}) +{\rm pol}_2 H_2(x_n,x_{n+1})+{\rm pol}_2 H_2(x_{n-1},x_{n+1}) +H_1(x_{n+1}). 
 \qquad 
\end{eqnarray}
Subtracting the latter two equations (taking into account the skew-symmetry of $J^{-1}$) leads to
\begin{eqnarray}\label{eq aux Ham 3}
\lefteqn{\tfrac{1}{\epsilon}\langle x_{n+1}, J^{-1}x_n\rangle - \tfrac{1}{\epsilon}\langle x_n, J^{-1}x_{n-1}\rangle}\nonumber \\
 & = & {\rm pol}_2 H_2(x_n,x_{n+1})-{\rm pol}_2 H_2(x_{n-1},x_n) +H_1(x_{n+1})-H_1(x_{n-1}). 
 \qquad 
\end{eqnarray}
This is equivalent to \eqref{eq d Ham int 1} being a conserved quantity.
\end{proof}

{\bf Discussion.}
\smallskip

1) It is not very common to express conserved quantities of a  first order difference equation in terms of more than one iterate. To avoid misconceptions, we stress that the statement that $H_\epsilon(x_n,x_{n+1})$ is a conserved quantity of the difference equation \eqref{eq Ham Kahan sep} means that 
$$
H_\epsilon(x,f_\epsilon(x))=H_\epsilon(f_\epsilon(x),f^2_\epsilon(x)).
$$
In other words, $H_\epsilon(x,f_\epsilon(x))$ is an integral of motion of the map $f_\epsilon$. It is in this latter form that the integral has been found in \cite{CMOQ1}. Earlier examples of expressions of conserved quantities of Kahan discretizations in terms of more than one iterate have been found in \cite{PPS2, PS}. 

2)  If $H(x)$ is homogeneous of degree 3, we get an especially simple conserved quantity $\epsilon H_\epsilon(x_n)=
\langle x_n\,,J^{-1}x_{n+1}\rangle$. This particular result was found previously in \cite{CMOQ3} as a special case of a more general statement for discretization by polarization.

3) It is instructive to look at the continuous time counterpart of this result. We derive, by Euler's theorem on homogeneous functions:
$$
\langle x, J^{-1}\dot x\rangle =\langle x,\nabla H(x)\rangle =3H_3(x)+2H_2(x)+H_1(x).
$$
As a consequence, the quantity
$$
\langle x, J^{-1}\dot x\rangle +H_2(x)+2H_1(x)
$$
is an integral of motion (equals $3H(x)$). In particular, if $H(x)$ is homogeneous of degree 3, we get a simple expression $\langle x,J^{-1}\dot x\rangle$ for the integral of motion.
\medskip

\textbf{Example.} Take $d=2$, $x=\begin{pmatrix} q \\p\end{pmatrix}$, $J=\begin{pmatrix} 0 & 1 \\ -1 & 0\end{pmatrix}$, so that 
$J^{-1}=\begin{pmatrix} 0 & -1 \\ 1 & 0\end{pmatrix}$, and set
\begin{eqnarray}
H_3(q,p) & = & a_{30}q^3+a_{21}q^2p+a_{12}qp^2+a_{03}p^3,\\
H_2(q,p) & = & a_{20}q^2+a_{11}qp+a_{02}p^2,\\
H_1(q,p) & = & a_{10}q+a_{01}p.
\end{eqnarray}
Thus, equations of motion \eqref{eq Ham system} read
\begin{eqnarray}
\dot q & = & a_{21}q^2+2a_{12}qp+3a_{03}p^2+a_{11}q+2a_{02}p+a_{01}, \\
\dot p & = & -3a_{30}q^2-2a_{21}qp-a_{12}p^2-2a_{20}q-a_{1}p-a_{10}, 
\end{eqnarray}
while their Kahan discretization reads
\begin{eqnarray}
(q_{n+1}-q_n)/ \epsilon & = & a_{21}q_nq_{n+1}+a_{12}(q_np_{n+1}+p_nq_{n+1})+3a_{03}p_np_{n+1}\nonumber\\
&& +\tfrac{1}{2}a_{11}(q_n+q_{n+1})+a_{02}(p_n+p_{n+1})+a_{01}, \label{eq d Ham q}\\
(p_{n+1}-p_n)/\epsilon & = & -3a_{30}q_nq_{n+1}-a_{21}(q_np_{n+1}+p_nq_{n+1})-a_{12}p_np_{n+1}\nonumber\\
&& -a_{20}(q_n+q_{n+1})-\tfrac{1}{2}a_{11}(p_n+p_{n+1})-a_{10}. \label{eq d Ham p}
\end{eqnarray}
Conserved quantity \eqref{eq d Ham int 1} takes the form
\begin{eqnarray}
H_\epsilon(q_n,p_n,q_{n+1},p_{n+1}) & = & \tfrac{1}{ \epsilon}(p_nq_{n+1}-q_np_{n+1})\nonumber\\
&& +a_{20}q_nq_{n+1}+\tfrac{1}{2}a_{11}(p_nq_{n+1}+q_np_{n+1})+a_{02}p_np_{n+1}\nonumber\\
&& +a_{01}(p_n+p_{n+1})+a_{10}(q_n+q_{n+1}). 
\end{eqnarray}
If $H(q,p)$ is homogeneous of degree 3, we get a quite simple conserved quantity $\epsilon H_\epsilon=p_nq_{n+1}-q_np_{n+1}$.
The continuous time limit of $H_\epsilon$ is the expression
$$
p\dot q-q\dot p +H_2(q,p)+2H_1(q,p),
$$
which is an integral of motion (equals $3H(q,p)$). In particular, if $H(q,p)$ is homogeneous of degree 3, we get a simple ``Wronskian'' expression $p\dot q-q\dot p$ for the integral of motion.

\section{Second order Hamiltonian systems with a quartic potential}
\label{sect 2nd order}

Consider a Hamiltonian system
\begin{equation}\label{eq Ham system 2nd order}
\ddot{x}=-K\nabla W, 
\end{equation}
where $W:\bbR^{d}\to \bbR$ is a polynomial of degree 4, and $K$ is a symmetric $d\times d$ matrix. This system admits an integral of motion
\begin{equation}\label{eq int 2nd order}
H(x,\dot x)=\frac{1}{2}\langle \dot x, K^{-1}\dot x\rangle +W(x).
\end{equation}
The right-hand side of equation \eqref{eq Ham system 2nd order} is of degree 3, and we consider the corresponding discretization by polarization, see \eqref{eq 2nd order discrete}.

\begin{theorem}\label{th Ham syst 2nd order} 
Separate the potential $W$ into homogeneous parts of degrees 4, 3, 2, and 1:
\begin{equation}
W(x)=W_4(x)+W_3(x)+W_2(x)+W_1(x).
\end{equation}
Then the following quantity is a conserved quantity of the difference equation \eqref{eq 2nd order discrete}:
\begin{eqnarray}\label{eq d Ham 2nd order int}
\lefteqn{H_\epsilon(x_{n-1},x_n,x_{n+1})  = \tfrac{1}{ \epsilon^2}\,\big(\langle x_{n-1},K^{-1}x_{n}\rangle -2\langle x_{n-1},K^{-1}x_{n+1}\rangle
+\langle x_{n},K^{-1}x_{n+1} \rangle\big)}
\nonumber\\
&& +{\rm pol}_3W_3(x_{n-1},x_n,x_{n+1})+2\,{\rm pol}_3W_2(x_{n-1},x_n,x_{n+1})+3\,{\rm pol}_3W_1(x_{n-1},x_n,x_{n+1}).\qquad\; \nonumber\\
\end{eqnarray}
\end{theorem}
\begin{proof}
We are dealing with the following difference equation:
\begin{eqnarray}\label{eq dHam 2nd order}
K^{-1}(x_{n+1}-2x_n+x_{n-1})/\epsilon^2 & = & -{\rm pol}_3 \nabla W_4(x_{n-1},x_n,x_{n+1})-{\rm pol}_3 \nabla W_3(x_{n-1},x_n,x_{n+1})\nonumber\\
 && -{\rm pol}_3 \nabla W_2(x_{n-1},x_n,x_{n+1})-{\rm pol}_3 \nabla W_1(x_{n-1},x_n,x_{n+1}).\nonumber\\
\end{eqnarray}
Take the scalar product of this equation with $x_{n+2}$:
\begin{eqnarray} \label{eq d Ham 2nd order aux1}
\lefteqn{\big(\langle x_{n+2}, K^{-1}x_{n+1}\rangle -2\langle x_{n+2}, K^{-1}x_n\rangle +\langle x_{n+2},K^{-1}x_{n-1}\rangle\big)/ \epsilon^2}\nonumber\\
&= &    -4\,{\rm pol}_4W_4(x_{n-1},x_n,x_{n+1},x_{n+2})\nonumber\\
& & -{\rm pol}_3 W_3(x_{n-1},x_n,x_{n+2})-{\rm pol}_3W_3(x_{n-1},x_{n+1},x_{n+2})-{\rm pol}_3W_3(x_n,x_{n+1},x_{n+2})\nonumber\\
& & -\frac{2}{3}\big({\rm pol}_2W_2(x_{n-1},x_{n+2})+{\rm pol}_2W_2(x_n,x_{n+2})+{\rm pol}_2W_2(x_{n+1},x_{n+2})\big)\nonumber\\
& & -W_{1}(x_{n+2}). 
\end{eqnarray}
Here we used Euler's theorem on homogeneous functions and have taken into account formulas \eqref{eq pol3 Q} for the quadratic form $\nabla W_3$ and \eqref{eq pol3 L} for the linear form $\nabla W_2$, and that $\nabla W_1$ is a constant vector. Similarly, take the scalar product of the shifted equation \eqref{eq dHam 2nd order} (i.e., $n\to n+1$), by $x_{n-1}$:
\begin{eqnarray}  \label{eq d Ham 2nd order aux2}
\lefteqn{\big(\langle x_{n-1}, K^{-1}x_{n+2}\rangle -2\langle x_{n-1}, K^{-1}x_{n+1}\rangle +\langle x_{n-1},K^{-1}x_{n}\rangle\big)/ \epsilon^2}\nonumber\\
&= &    -4\,{\rm pol}_4W_4(x_{n-1},x_n,x_{n+1},x_{n+2})\nonumber\\
& & -{\rm pol}_3 W_3(x_{n-1},x_n,x_{n+1})-{\rm pol}_3W_3(x_{n-1},x_{n},x_{n+2})-{\rm pol}_3W_3(x_{n-1},x_{n+1},x_{n+2})\nonumber\\
& & -\frac{2}{3}\big({\rm pol}_2W_2(x_{n-1},x_{n})+{\rm pol}_2W_2(x_{n-1},x_{n+1})+{\rm pol}_2W_2(x_{n-1},x_{n+2})\big)\nonumber\\
& & -W_{1}(x_{n-1}).
\end{eqnarray}
Subtracting \eqref{eq d Ham 2nd order aux2} from \eqref{eq d Ham 2nd order aux1} leads to:
\begin{eqnarray} \label{eq d Ham 2nd order aux3}
\lefteqn{\big(\langle x_{n+2}, K^{-1}x_{n+1}\rangle -2\langle x_{n+2}, K^{-1}x_n\rangle +2\langle x_{n-1}, K^{-1}x_{n+1}\rangle -\langle x_{n-1},K^{-1}x_{n}\rangle\big)/ \epsilon^2}\nonumber\\
& = & -{\rm pol}_3W_3(x_n,x_{n+1},x_{n+2})+ {\rm pol}_3 W_3(x_{n-1},x_n,x_{n+1})\nonumber\\
& & -\frac{2}{3}\big({\rm pol}_2W_2(x_n,x_{n+2})+{\rm pol}_2W_2(x_{n+1},x_{n+2})-{\rm pol}_2W_2(x_{n-1},x_{n+1})-{\rm pol}_2W_2(x_{n-1},x_{n})\big)\nonumber\\
& & -W_{1}(x_{n+2})+W_1(x_{n-1}). 
\end{eqnarray}
This is equivalent to the following expression being a conserved quantity:
\begin{eqnarray} \label{eq d Ham 2nd order aux4}
\lefteqn{\big(\langle x_{n-1},K^{-1}x_{n}\rangle-2\langle x_{n-1}, K^{-1}x_{n+1}\rangle +\langle x_{n}, K^{-1}x_{n+1}\rangle\big)/ \epsilon^2}\nonumber\\
&  & + {\rm pol}_3 W_3(x_{n-1},x_n,x_{n+1})\nonumber\\
& & +\frac{2}{3}\big({\rm pol}_2W_2(x_{n-1},x_{n})+{\rm pol}_2W_2(x_{n-1},x_{n+1})+{\rm pol}_2W_2(x_{n},x_{n+1})\big)\nonumber\\
& & +W_1(x_{n-1})+W_1(x_n)+W_1(x_{n+1}). 
\end{eqnarray}
But this is the same as \eqref{eq d Ham 2nd order int}.
\end{proof}

{\bf Discussion.}
\smallskip

1) Of course, in order to consider \eqref{eq d Ham 2nd order int} as a function of $(x_{n-1},x_n)$, one has to substitute on the right-hand side the rational expression of $x_{n+1}$ through $(x_{n-1},x_n)$, which follows from \eqref{eq dHam 2nd order}.

2)  If $W(x)$ is homogeneous of degree 4, we get an especially simple conserved quantity: 
$$
\epsilon^2 H_\epsilon(x_{n-1},x_n,x_{n+1})  = \langle x_{n-1},K^{-1}x_{n}\rangle-2\langle x_{n-1},K^{-1}x_{n+1}\rangle+
\langle x_{n},K^{-1}x_{n+1} \rangle.
$$

3) It is instructive to look at the continuous time counterpart of this result. The continuos limit of the expression on the right-hand side of \eqref{eq d Ham 2nd order int} (performed according to $x_n=x$, $x_{n\pm 1}=x\pm\epsilon \dot x+\tfrac{\epsilon^2}{2}\ddot x+O(\epsilon^3)$) equals
$$
2\langle \dot x, K^{-1}\dot x\rangle-\langle x, K^{-1}\ddot x \rangle+W_3(x)+2W_2(x)+3W_1(x). 
$$
By virtue of equations of motion \eqref{eq Ham system 2nd order}, this equals
$$
2\langle \dot x, K^{-1}\dot x\rangle+\langle x, \nabla W(x) \rangle+W_3(x)+2W_2(x)+3W_1(x)
$$
and,  by Euler's theorem on homogeneous functions, we find:
\begin{eqnarray*}
& & =2\langle \dot x, K^{-1}\dot x\rangle+\big(4 W_4(x)+3W_3(x)+2W_2(x)+W_1(x)\big)+W_3(x)+2W_2(x)+3W_1(x),\\
& & =2\langle \dot x, K^{-1}\dot x\rangle+4 W(x),
\end{eqnarray*}
which is an integral of motion $4H(x,\dot x)$, see \eqref{eq int 2nd order}.
\medskip

\textbf{Example.} We take $d=1$, $K=1$, and set
\begin{equation}
W(x)=\tfrac{1}{4}a_4x^4+\tfrac{1}{3}a_3x^3+\tfrac{1}{2}a_2x^2+a_1x.
\end{equation}
Thus, equations of motion \eqref{eq Ham system 2nd order} read
\begin{equation}
\ddot x =-a_4x^3-a_3x^2-a_2x-a_1,
\end{equation}
while their polarization discretization reads
\begin{eqnarray}\label{eq d Ham 2nd order scalar}
(x_{n+1}-2x_n+x_{n-1})/ \epsilon^2 & = & -a_4 x_{n-1}x_nx_{n+1} -\tfrac{1}{3}a_3(x_{n-1}x_n+x_{n-1}x_{n+1}+x_nx_{n+1})\nonumber\\
& & -\tfrac{1}{3}a_{2}(x_{n-1}+x_n+x_{n+1})-a_{1}. 
\end{eqnarray}
The following is a conserved quantity for the map $(x_{n-1},x_n)\mapsto (x_n,x_{n+1})$:
\begin{eqnarray}\label{eq d Ham 2nd order int scalar}
\lefteqn{H_\epsilon(x_{n-1},x_n,x_{n+1})  = \tfrac{1}{ \epsilon^2}\,\big(x_{n-1}x_{n}-2x_{n-1}x_{n+1}+x_nx_{n+1}\big)}
\nonumber\\
&& +\tfrac{1}{3}a_3x_{n-1}x_nx_{n+1}+\tfrac{1}{3}a_2(x_{n-1}x_n+x_{n-1}x_{n+1}+x_nx_{n+1})+a_1(x_{n-1}+x_n+x_{n+1}).\qquad\quad
\end{eqnarray}
Upon expressing $x_{n+1}$ through $(x_{n-1},x_n)$ by virtue of equation \eqref{eq d Ham 2nd order scalar}, this coincides with the integral found in \cite{HQ}.
 
\section{Conclusion}
It is hoped that the algebraic approach to derivation of integrals of motion for the discrete time versions of Hamiltonian systems obtained by polarization will further stimulate the development of this fascinating area, towards an ultimate understanding of all the miraculous results discovered up to this day and yet to be discovered.

\subsection*{Acknowledgements}

This research is supported by the DFG Collaborative Research Center TRR 109 ``Discretization in Geometry and Dynamics''.

\label{lastpage}
\end{document}